\documentclass[12pt]{article}
\usepackage[english]{babel}
\usepackage{amsthm}
\usepackage{amsfonts}
\usepackage{amsmath,amssymb,fullpage,graphicx,natbib,color}
\let\hat\widehat
\let\tilde\widetilde

\usepackage{sectsty}
\sectionfont{\large}
\subsectionfont{\small}

\usepackage[compact]{titlesec}
\titlespacing{\section}{0pt}{*2}{*0}
\titlespacing{\subsection}{0pt}{*1}{*0}
\titlespacing{\subsubsection}{0pt}{*0}{*0}

\newtheorem{theorem}{Theorem}
\newtheorem{lemma}[theorem]{Lemma}

\newtheorem{corollary}[theorem]{Corollary}
\newtheorem{proposition}[theorem]{Proposition}
\newtheorem{remark}[theorem]{Remark}
\newtheorem{definition}[theorem]{Definition}

\DeclareMathOperator*{\esssup}{ess\,sup}
\newcommand{\E}{\mathbb{E}}
\newcommand{\ind}{1{\hskip -2.5 pt}\hbox{I}}

\newenvironment{enum}{
\begin{enumerate}
  \setlength{\itemsep}{1pt}
  \setlength{\parskip}{0pt}
  \setlength{\parsep}{0pt}
}{\end{enumerate}}

\begin{document}

\begin{center}
{\Large\bf Distribution Free Prediction Bands}
\end{center}

\begin{center}
{Jing Lei and Larry Wasserman}\\
{Carnegie Mellon University}\\
\today
\end{center}

\begin{abstract}
We study distribution free, nonparametric prediction bands
  with a special focus on their finite sample behavior.
First we investigate and develop different notions of
finite sample
coverage guarantees.
Then we give a new prediction band estimator
by combining the idea
of ``conformal prediction'' \citep{VovkNG09}
with nonparametric conditional density estimation.
The proposed estimator, called COPS (Conformal
Optimized Prediction Set), always has finite sample guarantee
in a stronger sense than the original conformal prediction
estimator.
Under regularity conditions the estimator
converges to an oracle band at a minimax optimal rate.
A fast approximation algorithm and a data driven method for
selecting the bandwidth are developed.
The method is illustrated first in simulated data. Then, an application
shows that the proposed method gives
desirable prediction intervals in an automatic
way, as compared to the classical linear regression modeling.
\end{abstract}

\section{Introduction}

Given observations $(X_i,Y_i)\in \mathbb R^{d}\times \mathbb R^1$ for $i=1,...,n$,
we want to
predict $Y_{n+1}$ given future predictor $X_{n+1}$.
Unlike typical nonparametric regression methods, our goal is not to produce a point prediction.
Instead, we construct a prediction interval $C_n$
that contains $Y_{n+1}$ with probability at least $1-\alpha$.
More precisely,
assume that
$(X_1,Y_1),\cdots (X_{n+1},Y_{n+1})$ are iid observations from some
distribution $P$.
We construct,  from the first $n$ sample points, a set-valued function
\begin{equation}
C_n(x)\equiv C_n(X_1,Y_1,\ldots,X_n,Y_n,x)\subseteq \mathbb R^1
\end{equation}
such that the next response variable $Y_{n+1}$ falls inside
$C_n(X_{n+1})$ with a certain level of confidence.
The collection of prediction sets
$C_n=\left\{ C_n(x):\ x\in\mathbb{R}^d\right\}$
forms a
\emph{prediction band}.

The prediction set $C_n(x)$ depends on
the observed value $X_{n+1}=x$, which shall be interpreted
as the estimated set that $Y$ is likely to fall in,
given $X_{n+1}=x$.  This extends nonparametric regression
by providing a prediction set
for each $x$.  Such a prediction set
provides useful information
about the uncertainty.  The problem of
prediction intervals is well studied in the context
of linear regression, where prediction intervals
are constructed under linear and Gaussian assumptions
(see, \cite{DeGrootS12}, Theorem 11.3.6).
The Gaussian assumption can be relaxed using, for example,
quantile regression \citep{KoenkerH01}.
These linear model based methods usually have
reasonable finite sample performance. However,
the coverage is valid only when the linear (or other parametric)
regression
model is correctly specified.  On the other hand,
nonparametric methods
have the potential to work for any smooth distribution
(\cite{ruppert})
but only asymptotic results are available and
the finite sample behavior remains unclear.

Recently, \citet{VovkNG09} introduce a generic approach,
called \emph{conformal prediction},
to construct
valid, distribution free, sequential prediction sets.
When adapted to our setting, this yields
prediction bands with a
\emph{finite sample coverage guarantee}
in the sense that
\begin{equation}\label{eq::valid}
\mathbb{P}\left[Y_{n+1}\in C_n(X_{n+1})\right]\geq 1-\alpha\ \ \
{\rm for\ all\ } P,
\end{equation}
where $\mathbb{P}=P^{n+1}$
is the joint measure of
$(X_1,Y_1),\cdots (X_{n+1},Y_{n+1})$.
However,
the conditional coverage and statistical efficiency
of such bands are not investigated.

In this paper we extend
the results in \cite{VovkNG09} and study
conditional coverage as well as efficiency.
We show that although finite sample coverage defined
in (\ref{eq::valid}) is a desirable
property, this is not enough to guarantee good prediction bands.
We argue that the finite sample coverage given by (\ref{eq::valid})
should be interpreted as \emph{marginal
coverage}, which is different from (in fact, weaker than) the
\emph{conditional coverage} as usually sought in prediction
problems.  Requiring only marginal validity may
lead to unsatisfactory estimation even in very simple cases.
As a result, a good estimator must satisfy something
more than
marginal coverage.  A natural criterion would
be conditional
coverage.
However, we prove that conditional coverage
is impossible to achieve with a finite sample.
As an alternative solution, we develop a new notion, called
\emph{local
validity}, that interpolates between marginal and
conditional
validity, and is achievable with a finite sample.
This notion leads to our proposed estimator: COPS (Conformal
Optimized Prediction Set).
We also show that when the sample size goes to infinity,
under regularity conditions,
the locally valid prediction band
given by COPS can give arbitrarily
accurate conditional coverage, leading to
an asymptotic conditional coverage guarantee.

Another contribution of this paper is the study of
\emph{efficiency} in the context of nonparametric
prediction bands.  Roughly speaking, efficiency
requires a prediction band to be small while maintaining
the desired probability coverage in the sense of (\ref{eq::valid}).
We study the efficiency of our estimator by measuring
its deviation from an {\em oracle band},
the band one should use
if the joint distribution $P$ were known.
We also give a
minimax lower bound on the estimation
error so that the efficiency of our method is indeed
minimax rate optimal over a certain class of smooth
distributions.

To summarize, the method given in this paper
is the first one with both finite sample
(marginal and local) coverage,
asymptotic conditional coverage,
and an explicit rate for asymptotic efficiency.
The finite sample marginal and local validity is
distribution
free: no assumptions on $P$ are required; $P$ need not
even have a density.
Asymptotic conditional validity and efficiency are closely
related
and rely on some standard regularity conditions on the
density.
Furthermore, all tuning parameters are completely data-driven.

The problem of
constructing prediction bands resembles
that of nonparametric confidence band estimation for the
regression function $m(x)=\mathbb E(Y|X=x)$.
However, these are two different inference problems.
First note that non-trivial, distribution-free
confidence bands for the regression function
$m(x) = \mathbb{E}(Y|X=x)$
do not exist \citep{Low,GW}.
On the other hand, in this paper we show that
consistent prediction bands estimation is
possible under mild regularity conditions.
Hence there is a distinct difference between confidence bands
for the regression function
and prediction bands.

{\em Prior Work On Nonparametric Prediction Bands.}
The usual nonparametric prediction interval takes the form
\begin{equation}
\hat{m}(x) \pm z_{\alpha/2}\, \sqrt{\hat\sigma^2 + s^2}
\end{equation}
where $\hat m$ is some nonparametric regression estimator,
$\hat\sigma^2$ is an estimate of ${\sf Var}(Y|X)$,
$s$ is an estimate of the standard error of $\hat m$
and $z_{\alpha/2}$ is either a Normal quantile
or a quantile determined by bootstrapping.
See, for example, Section 6.2 of
\cite{ruppert},
Section 2.3.3 of \cite{loader}
and Chapter 5 of \cite{Fan}.
The assumption of constant variance can be relaxed; see, for example,
\cite{SJOS:SJOS254}.
Other related work includes
\cite{RSSB:RSSB308} on bootstrapping,
\cite{dav1987} on variance estimation and
\cite{car1991} on transformation approaches.
However, none of these methods yields prediction bands with
distribution free, finite sample
 validity.
Furthermore, these methods always produce a prediction set
in the form of an interval which, as we shall see,
may not be optimal.
In fact, we are not aware of any paper that
provides distribution free finite sample prediction bands
with asymptotic optimality properties as we provide in this paper.
The only paper we know of that provides
finite sample marginal validity is
the very interesting paper by \cite{VovkNG09}.
However, that paper focuses on linear predictors
and does not address efficiency or conditional validity.

{\em Outline.}
In Section \ref{section::validity}
we introduce various notions of validity
and efficiency.
In Section \ref{sec::method} we introduce our methods
for prediction bands: the COPS estimator.
We study the large sample and minimax results of the method
in Section \ref{sec:rates}.
We discuss bandwidth selection in
Section \ref{sec:bandwidth}.
Section \ref{sec::examples} contains
some examples.
Finally, concluding remarks are in Section
\ref{sec::conclusion}.

\section{Marginal, Conditional, and Local Validity}\label{sec:interpretation}
\label{section::validity}

\subsection{Marginal Validity and Prediction Sets}

Prediction bands are an extension
of nonparametric prediction sets
(also called tolerance regions).
Suppose we observe $n$ iid copies
$Z_1,\ldots, Z_n$
of a random vector $Z$ with distribution
$P$ and we want a
set $T_n\subseteq R^d$ such that
$\mathbb P\left[Z_{n+1}\in T_n\right]\ge 1-\alpha$
for all $P$.
Let $Z_i=(X_i,Y_i)$.
Since
the probability
statement in (\ref{eq::valid})
is over the joint distribution of
$(X_1,Y_1),\ldots, (X_{n+1},Y_{n+1})$,
it is equivalent to
\begin{equation}\label{eq:valid_equiv}
\mathbb P\left[(X_{n+1},Y_{n+1})\in C_n\right]
\ge 1-\alpha,~{\rm for~all~}P.
\end{equation}
That is, equation (\ref{eq:valid_equiv}) is exactly
the definition of a prediction set for the joint
distribution $(X,Y)$.
As a result, any prediction set for the joint distribution
provides a solution, with finite sample coverage,
to the prediction band problem.
In this subsection we pursue this point further.
In the following subsections we consider
improvements.

The study of
prediction sets dates back to \cite{Wilks41},
\cite{Wald43}, and \cite{Tukey47}.  More recently,
the research on prediction sets has focused
on finding statistically efficient estimators in
multivariate cases
\citep{ChatterjeeP80,DiBucchianicoEM01,LiL08}.
\cite{LeiRW11} study distribution free, finite sample valid
and efficient
estimator of
prediction sets.
A thorough introduction to
prediction sets can be found in \cite{KrishnamoorthM09}.

There are many different methods to construct prediction sets.
A common measure of efficiency  is
the Lebesgue measure and the optimal prediction set
is the one with smallest Lebesgue measure among all sets with
the desired coverage level.
It is well-known that the optimal prediction set at level
$1-\alpha$ (optimal refers to the one having smallest Lebesgue measure) is
an upper level set of the joint density:
\begin{equation}
\label{eq:oracle_toler}
C^{(\alpha)}=\Bigl\{(x,y):p(x,y)\ge t^{(\alpha)}\Bigr\},
\end{equation}
where $t^{(\alpha)}$ is chosen such hat
$P(C^{(\alpha)})=1-\alpha$.
As illustrated in
the following example, an optimal joint prediction
can lead to an unsatisfactory prediction band.

\begin{figure}
\begin{center}
\includegraphics[height = 6 cm, width=14cm]{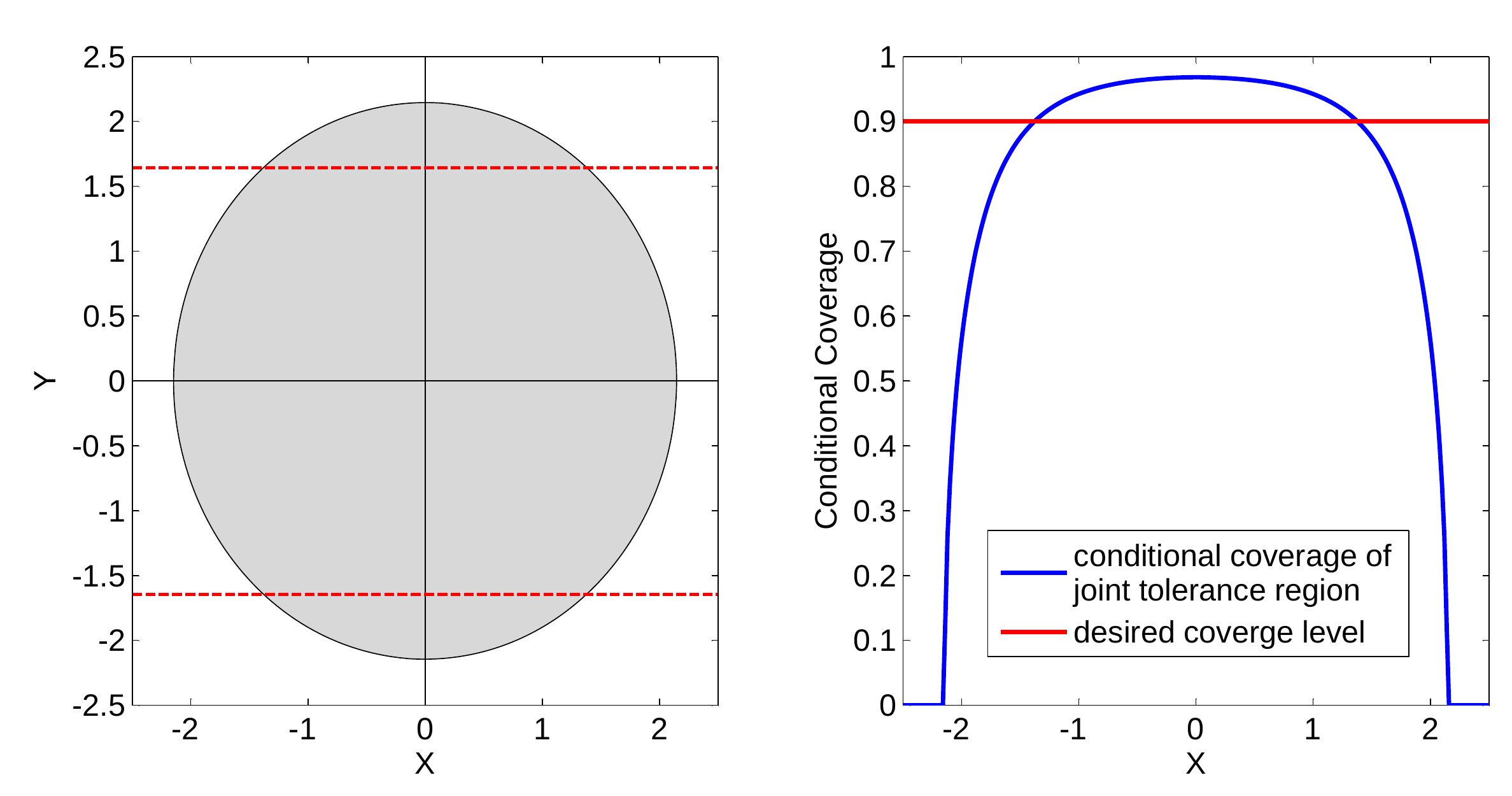}
\caption{Joint prediction set and pointwise conditional
coverage for bivariate independent Gaussian.
Left panel: the gray area is the optimal (with smallest
Lebesgue  measure) prediction set with coverage 0.9, the two red lines are the
upper and lower 5\% quantiles of the marginal distribution of
$Y$.  Right panel: the blue curve plots $P(Y\in C(x)|X=x)$
against $x$; the red line is the desired coverage level 0.9.}
\label{fig:indep_normal}
\end{center}
\end{figure}

Figure \ref{fig:indep_normal} shows the case of a bivariate independent
Gaussian. According to (\ref{eq:oracle_toler}), when $X,Y$ are independent standard normals,
the level set for
any $C^{(\alpha)}$ is a circle centered at the origin as described by the
gray area in the left panel of Figure \ref{fig:indep_normal}.  But intuitively
since observing $X$ provides no information about $Y$, the best
prediction band at level $\alpha$ should be
$C(x)=[-z_{\alpha/2}, z_{\alpha/2}]$, for all $x$, where $z_{\tau}$ is the
$\tau$-th upper quantile of standard normal.  This
band is the set between two red
dashed lines in the left panel of
Figure \ref{fig:indep_normal} for $\alpha=0.1$.

In prediction, another important notion of coverage is the conditional
coverage $P(Y\in C(x)|X=x)$.
The pointwise conditional coverage $P(Y\in C(x)|X=x)$ is plotted
in the right panel of Figure \ref{fig:indep_normal} for the
joint prediction set (blue curve).  We see that the ``optimal''
joint prediction set tends to overestimate the set when $x$
is in the high density area and to underestimate for low density $x$.
Let us now consider conditional validity in more detail.

\subsection{Conditional Validity}

Only requiring (\ref{eq::valid}) for prediction bands is
not enough.
We will refer to (\ref{eq::valid}) as
{\em marginal validity} or {\em joint validity}.
This is the type of validity used in \cite{ShaferV08}.
As illustrated in the example above, it
may be tempting to insist on a more stringent probability guarantee
such as
\begin{equation}\label{eq:cond_valid}
\mathbb{P}(Y_{n+1}\in C_n(x)|X_{n+1}=x) \geq 1-\alpha\ \ \
{\rm for \ all\ }P\ {\rm and\ }
{\rm almost\ all\ }x,
\end{equation}
which we call {\em conditional validity}.
If the joint distribution of $(X,Y)$ is known, one
can define an oracle band as the counterpart of (\ref{eq::valid}) for
conditionally valid bands:
\begin{equation}\label{eq:cond_oracle}
C_P(x)=\Bigl\{y:p(y|x)\ge t^{(\alpha)}(x)\Bigr\}
\end{equation}
where
$t^{(\alpha)}(x)$ satisfies
$$
\int \ind \left\{p(y|x)\ge t^{(\alpha)}(x)\right\} p(y|x)dy= 1-\alpha.
$$
We call
$C_P = \{C_P(x):\ x\in \mathbb{R}^d\}$ the {\em conditional oracle band}.
It is easy to prove that $C_P$ minimizes
$\mu[C(x)]$ for all $x$
among all
bands satisfying
$\inf_x P(Y\in C(x)|X=x)\ge 1-\alpha$.
Note that $C_P$ depends on $P$
but does not depend on the observed data.
For an estimator $\hat C$, {\em asymptotic efficiency} requires
$\hat C(x)$ be close to $C_P(x)$ uniformly over all $x$:
\begin{equation}
\label{eq:asymp_efficiency}
\sup_x\mu\Bigl[\hat C(x)\triangle C_P(x)\Bigr]
\stackrel{P}{\rightarrow}0.
\end{equation}
However, we will show that
there do not exist any prediction bands
$\hat C$ that satisfy both (\ref{eq:cond_valid}) and
(\ref{eq:asymp_efficiency}).  In fact, the following
claim, proved in Subsection \ref{app:pf_impossible},
is even stronger.

Let $P_X$ denote the marginal distribution of $X$
under $P$. A point $x$ is a {\em non-atom} for $P$ if
$x$ is in the support of $P_X$ and if
$P_X[B(x,\delta)]\to 0$
as $\delta\to 0$, where $B(x,\delta)$ is the Euclidean ball
centered at $x$ with radius $\delta$.
Let $N(P)$ denote the set of non-atoms.
We show that if $C_n$
is conditionally valid then the length of $C_n(x)$
is infinite for all $x\in N(P)$.

\begin{lemma}[Impossibility of non-trivial finite sample conditional validity]
\label{lem:impossible}
Suppose that an estimator
$C_n$ has $1-\alpha$ conditional validity.
For any $P$ and any $x_0\in N(P)$,
$$
\mathbb P\Biggl(\lim_{\delta\to 0} \esssup_{||x_0-x||\leq \delta}
\mu[C_n(x)] =\infty\Biggr) =1.
$$
\end{lemma}

Thus, non-trivial finite sample conditional validity
is impossible for continuous distributions.
We shall instead construct prediction bands
with an asymptotic version of
(\ref{eq:cond_valid}) together with
finite sample marginal validity.
We say that
$\hat C$ is {\em asymptotically conditionally valid} if
\begin{equation}\label{eq:asymp_cond_valid}
\sup_x \Bigl[ \mathbb{P}(Y_{n+1}\notin C_n(x)|X_{n+1}=x) - \alpha\Bigr]_+
\stackrel{P}{\to} 0
\end{equation}
as $n\to\infty$.
Here, the supremum is taken over the support of
$P_X$.  We note that
if the conditional density $p(y|x)$ is
uniformly bounded for all $(x,y)$, then
asymptotic conditional validity
is a consequence of asymptotic efficiency defined as in
(\ref{eq:asymp_efficiency}).

In Section \ref{sec::method}
we construct a prediction band that satisfies:
\begin{enum}
\item finite sample marginal validity,
\item asymptotic conditional validity and
\item asymptotic efficiency.
\end{enum}

Our method is based on the notion of
\emph{local validity}, which naturally interpolates between
marginal and conditional validity.

\begin{definition}[Local validity]
\label{def:local_valid}
Let
$\mathcal A=\{A_j:j\ge 1\}$ be a partition of ${\rm supp}(P_X)$ such that
each $A_j$
has diameter at most $\delta$.
A prediction band $C_n$ is locally valid with respect
to $\cal A$ if
\begin{equation}
\mathbb
P(Y_{n+1}\in C_n(X_{n+1})|X_{n+1}\in A_j) \geq 1-\alpha,
~~{\rm for\ all\ }j~{\rm and~all~}P.
\end{equation}
\end{definition}
\textbf{Remark.}
From the insight of Lemma \ref{lem:impossible}, it is possible to construct
finite sample locally valid prediction sets
because $X\in A_j$ is an event with positive
probability and hence repeated observations
are available.

\textbf{Remark.}
Consider the limiting case of
$\delta\rightarrow\infty$,
which can be thought as having $A_1={\rm supp}(P_X)$, and
local
validity becomes marginal validity.  On the other hand,
in the extremal case $\delta\rightarrow0$, $A_j$
shrinks
to a single point $x\in \mathbb R^d$, and local
validity approximates conditional validity. We also
note that local validity is stronger than marginal
validity but weaker than conditional validity.
We state the following proposition whose proof is
elementary and omitted.

\begin{proposition}
\label{pro:relation}
If $C$ is conditionally valid, then it is also
locally valid for any partition $\cal A$.
If $C$ is locally valid for some partition $\cal A$,
then it is also marginally valid.
\end{proposition}

The relationship between local validity and
asymptotic conditional validity is more complicated
and is one of the technical contributions of this paper.
In Section \ref{sec::method}
we construct a specific class of
locally valid bands. In Theorem
\ref{thm:main_consist} of Section
\ref{sec:rates} we show that under
mild regularity conditions, these
bands are also asymptotically conditionally valid.
To summarize, if $C$ is locally valid then it is also
marginally valid.
And under regularity conditions, it can also be asymptotically conditionally valid.
See Figure \ref{fig::validity}.

\begin{figure}
\begin{center}
\includegraphics[scale=.5]{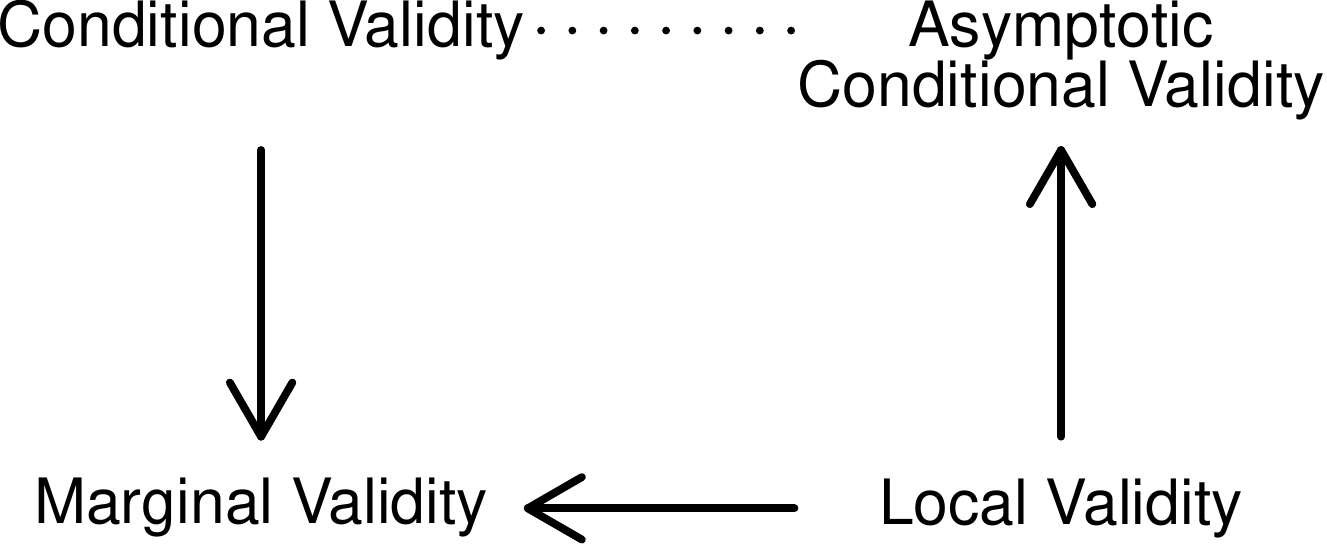}
\end{center}
\caption{Relationship between different types of validity.}
\label{fig::validity}
\end{figure}

How can we construct finite sample locally valid prediction
bands?  A straightforward approach is to apply the
method developed in \citet{LeiRW11} to
$P_j\equiv
\mathcal L(X,Y|X\in A_j)$, the joint distribution of $(X,Y)$ conditional
on the event $X\in A_j$.  Note that we are mostly
interested in the case $\max_j{\rm diam}(A_j)\rightarrow 0$,
therefore the marginal density of $X$ within $P_j$ becomes
increasingly close to uniform.  Therefore,
the approach can be simplified to finding
$C_j\in \mathbb R^1$, such that
$P(Y\in C_j|X\in A_j)\ge 1-\alpha$. This approach is
detailed in Section \ref{sec::method} and analyzed
in Section \ref{sec:rates}.

\section{Methodology}
\label{sec::method}

\subsection{A Marginally Valid Prediction Band}
\label{sec::no-covariates}

We start by recalling the construction of
joint prediction sets using kernel density
together with the idea of conformal prediction, as
described in \citet{LeiRW11},
using the idea of {\em conformal prediction}
developed in
\cite{ShaferV08}, \cite{VovkGS08b} and
\cite{VovkNG09}. This approach is
shown to have finite sample validity as well as
asymptotic efficiency under regularity conditions.
Suppose we observe
$$
Z_1,\ldots, Z_n\sim P
$$
and we want a prediction set for $Z_{n+1}$.
The idea is to test
$H_0: Z_{n+1}=z$ for each $z$ and
then invert the test.
Specifically, for any $z$ let
$\hat p_n^z(\cdot)$
be a density estimator based on
the {\em augmented data}
${\sf aug}(\mathbf Z;z)=(Z_1,\ldots,Z_n,z)$.
Define
$$
C_n \equiv C_n(Z_1,\ldots,Z_n) = \left\{z:\ \pi_n(z)\geq \alpha\right\}
$$
where
$$
\pi_n(z) = \frac{1}{n+1}\sum_{i=1}^{n+1} \ind(\sigma_i(z)\leq \sigma_{n+1}(z))
$$
is the p-value for the test,
$\sigma_i(z) = \hat{p}_n^z(Z_i)$ for $i=1,\ldots, n$
and
$\sigma_{n+1}(z) = \hat{p}_n^z(z)$.
The statistic $\sigma_i$ is an example of a {\em conformity measure}.
More generally, a conformity measure
$\sigma_i(z)=\sigma({\sf aug}(\mathbf Z, z),Z_i)$ indicates
how well a data point $Z_i$ agrees with
the augmented data set ${\sf aug}(\mathbf Z,z)$.  In principle
$\sigma(\cdot,\cdot)$ can be any function but usually it makes sense to
use the fitted residual or likelihood at $Z_i$ with respect
to a model estimated from ${\sf aug}(\mathbf Z, z)$.

The intuition for $C_n$
is the following.
Fix an arbitrary value $z$.
To test
$H_0: Z_{n+1}=z$
we use
the heights of the density estimators $\sigma_i(z)=\hat p_n^z(Z_i)$
as a test-statistic.
(Note that $\sigma_1,\ldots,\sigma_{n+1}$
are functions of ${\sf aug}(\mathbf Z,z)$.)
Under $H_0$, the ranks of the
$\sigma_i$ are uniform, because the joint distribution
of $(Z_1,...,Z_n,Z_{n+1})$ does not change under
permutations.  Hence the vector
$(\sigma_1,...,\sigma_{n+1})$ is exchangeable.
Therefore, under $H_0$,
$\pi_n(z)$ is uniformly distributed over $[0,1]$
and is a valid p-value for the test.\footnote{More Precisely, it is sub-uniform due to
the discreteness.}
The set $C_n$ is obtained by inverting
the hypothesis test, that is, $C_n$ consists of all values $z$
that are not rejected by the test.
It then follows that
$\mathbb{P}(Z_{n+1}\in C_n)\geq 1-\alpha$ for all $P$.

In \cite{LeiRW11}, the density $\hat p_n^z$ is obtained from
kernel density estimators with bandwidth $h$.
\cite{LeiRW11} show that $\hat C^{(\alpha)}$
is also efficient meaning that
it is close to $C^{(\alpha)}$ with high probability
where $C^{(\alpha)}$ is the smallest set with
probability content $1-\alpha$ as defined in (\ref{eq:oracle_toler}).

Computing $\hat C^{(\alpha)}$ is expensive since
we need to find the
the p-value $\pi_n(z)$ for every $z$.
\cite{LeiRW11} proposed the following
approximation $C_n^+$ to $C_n$---called the sandwich approximation---
which avoids the augmentation step altogether but preserves finite sample
validity.
Let
$Z_{(1)},Z_{(2)},\ldots,$
denote the data ordered increasingly by
$\hat p(Z_i)$.
Let $j=\lfloor n\alpha\rfloor$ and
define
\begin{equation}\label{eq::sandwich-approx}
C_n^+ = \Biggl\{ z:\ \hat{p}(z) \ge \hat{p}(Z_{(j)}) - \frac{K(0)}{nh^d} \Biggr\}.
\end{equation}
\cite{LeiRW11} show that
$\hat C^{(\alpha)}\subseteq C_n^+$
and hence
$C_n^+$ also has finite sample validity.
Moreover,
$C_n^+$ has the same efficiency properties as $C_n$
if $h$ is chosen appropriately.
This result, known as the ``Sandwich Lemma'', provides a simple
characterization of the conformal prediction set $\hat C^{(\alpha)}$
in terms of the plug-in density level set. In this paper,
a specific version of the Sandwich Lemma
for the conditional density is stated in
Lemma \ref{lem:sandwich}.
Thus, using the sandwich approximation we get a fast method
for constructing a valid band, based on slicing the joint density.

Now let $Z=(X,Y)$.
The $x$-slices of the joint region for $Z$
define a marginally valid band.
Specifically,
let $K_x$ and $K_y$ be two kernel functions in $\mathbb R^d$ and $\mathbb R^1$,
respectively.
Consider the kernel density estimator: For any $(u,v)\in \mathbb R^{d}\times
\mathbb R^1$:
\begin{equation}
\hat p_{n;X,Y}(u,v)=\frac{1}{n}\sum_{i=1}^n\frac{1}{h_n^{d+1}}
K_x\left(\frac{u-X_i}{h_n}\right)K_y\left(\frac{v-Y_i}{h_n}\right).
\end{equation}

For any $(x,y)\in\mathbb R^d\times\mathbb R^1$,
let $({\bf X, Y})=(X_1,Y_1, \ldots ,X_n,Y_n)$ be the data set and
${\sf aug}({\bf X},{\bf Y}; (x, y))$ be
the augmented data with $X_{n+1}=x$ and $Y_{n+1}=y$.
Define $\hat p^{(x,y)}_{n;X,Y}$ be
the kernel density estimator from the augmented data:
\begin{equation}
\hat p_{n;X,Y}^{(x,y)}(u,v)=\frac{n}{n+1}\hat p_{n;X,Y}(u,v)+\frac{1}{(n+1)h_n^{d+1}}
K_x\left(\frac{u-x}{h_n}\right)K_y\left(\frac{v-y}{h_n}\right)\,.
\end{equation}
Define the conformity measure
\begin{equation}
\sigma_i(x,y):=
\hat p_{n;X,Y}^{(x,y)}(X_i,Y_i).
\end{equation}
and p-value
\begin{equation}
\pi_{i}=\frac{1}{n+1}\sum_{j=1}^{n+1} \ind(\sigma_j(x,y)\le \sigma_i(x,y))\,,
\quad{\rm for}\quad 1\le i\le n+1.
\end{equation}
Let $
\tilde\alpha = \lfloor (n+1)\alpha \rfloor/(n+1).
$
Since $(X_i,Y_i)_{i=1}^{n+1}$ are iid, by exchangeability, we have, for all $i$,
\begin{equation}\label{eq::exch}
\mathbb P(\pi_i \ge \tilde\alpha)\ge 1-\alpha.
\end{equation}
Define
$$
\hat C^{(\alpha)}(x)=\left\{y:\pi_{n+1}(x,y)\ge \tilde\alpha\right\},
$$
where
$\pi_{n+1}\equiv \pi_{n+1}\left[{\sf aug}({\bf X,Y}; (x,y))\right]$.
From (\ref{eq::exch}) we have:

\begin{lemma}
$\hat C^{(\alpha)}(x)$ is finite sample marginally valid:
$$
\mathbb{P}\left[Y_{n+1}\in \hat C^{(\alpha)}(X_{n+1})\right]\geq 1-\alpha\ \ \ {\rm for\ all\ }P.
$$
\end{lemma}

Now we use the sandwich approximation to the joint conformal region
for $(X,Y)$.
The resulting band $C_n^+(x)$ is obtained by fixing $X=x$
and taking slices of the joint region and
is then a marginally valid band.
See Algorithm 1.

\begin{figure}
\begin{center}
\fbox{\parbox{6in}{
\begin{center}
{\sf Algorithm 1. Sandwich Slicer Algorithm}
\end{center}
\begin{enum}
\item Let $\hat p(x,y)$ be the joint density estimator.
\item Let $Z_i = (X_i,Y_i)$ and let
$Z_{(1)},Z_{(2)},\ldots,$
denote the sample ordered increasingly by
$\hat{p}(X_i,Y_i)$.
\item Let $j=\lfloor n\alpha\rfloor$ and
define
\begin{equation}\label{eq::sandwich-approx2}
C_n^+(x) = \left\{ y:\ \hat{p}(x,y) \ge \hat{p}(X_{(j)},Y_{(j)}) - \frac{K_x(0)K_y(0)}{nh^{d+1}} \right\}.
\end{equation}
\end{enum}
}}
\end{center}
\end{figure}

To summarize: the band
given in
Algorithm 1
is marginally valid.
But it is not efficient nor does it satisfy
asymptotic conditional validity.
This leads to the subject of the next section.

\subsection{Locally Valid Bands}
\label{subsec:local_valid_construct}

Now we extend the idea of conformal prediction
to construct prediction bands with local validity.
These bands will also be
asymptotically efficient and
have
asymptotic conditional validity.
For simplicity of presentation, we assume that
${\rm supp}(P_X)=[0,1]^d$
where
${\rm supp}(P_X)$ denotes the support of $P_X$
and we consider partitions $\mathcal A=\{A_k,k\ge 1\}$
in the form of cubes
with sides of length $w_n$.
Let $n_k=\sum_{i=1}^n\mathbf 1(X_i\in A_k)$
be the histogram count.

Given a kernel function $K(\cdot):\mathbb R^1\mapsto \mathbb R^1$ and
another bandwidth $h_n$, consider the estimated
local marginal density of $Y$:
$$
\hat p(y|A_k)=\frac{1}{n_k h_n}
\sum_{i=1}^n\ind(X_i\in A_k)K\left(\frac{Y_i-y}{h_n}\right).
$$
The corresponding augmented estimate is, for
any $(x,y)\in A_k\times\mathbb R^1$,
\begin{equation}\label{eq:local_kernel_density_aug}
\hat p^{(x,y)}(v|A_k)=\frac{n_k}{n_k+1}\hat p(v|A_k)+\frac{1}{(n_k+1)h_n}K\left(
\frac{v-y}{h_n}\right).
\end{equation}
For any $(x,y)\in A_k\times \mathbb R^1$,
consider the following {\em local conformity rank}
\begin{equation}\label{eq:ldr}
\pi_{n,k}(x,y)=
\frac{1}{n_k+1}
\sum_{i=1}^{n+1}
\ind(X_i\in A_k)\ind\left[\hat p^{(x,y)}(Y_i|A_k)\le \hat
p^{(x,y)}(Y_{n+1}|A_k)\right]\,,\end{equation}
which can be interpreted as
the local conditional density rank.  It is easy to check
that the $\pi_{n,k}(x,y)$ has a sub-uniform distribution if
$(X_{n+1},Y_{n+1})=(x,y)$ is another independent sample from $P$.
Therefore, the band
\begin{equation}\label{eq:loc_val_band}
\hat C(x)=\{\pi_{n,k}(x,y)\ge \alpha\}\end{equation}
for $x\in A_k$
has finite sample local validity.

\begin{proposition}\label{pro:loc_valid_constr}
For $x\in A_k$, let
$\hat C(x)=\{y: \pi_{n,k}(x,y)\ge \alpha\}$,
where $\pi_{n,k}(x,y)$ is
defined as in
(\ref{eq:ldr}), then $\hat C(x)$
is finite sample locally valid and hence finite
sample marginally valid.
\end{proposition}
\begin{proof}
Fix $k$, let $\{i_1,...,i_{n_k}\}=\{i:1\le i\le
n,~X_i\in A_k\}$.
Let $(X_{n+1},Y_{n+1})\sim P$ be another
independent sample. Define $i_{n_k+1}=n+1$
and $\sigma_{i_{\ell}}=\hat p^{(x,y)}(Y_{i_\ell}|A_k)$
for all $1\le \ell\le n_{k}+1$.
Then conditioning on
the event
$X_{n+1}\in A_k$ and $(i_1,...,i_{n_k})$, the
sequence
$(\sigma_{i_1},...,\sigma_{i_{n_k}},\sigma_{i_{n_k+1}})$ is
exchangeable.
\end{proof}

We call $\hat C$ the Conformal Optimized Prediction Set
(COPS) estimator, where the word ``optimized'' stands for
the effort of minimizing the average interval length
$\E_X \hat C(X)$.

We give a fast approximation algorithm that is
analogous to Algorithm 1.  The resulting approximation
also satisfies finite sample local validity as well as
asymptotic efficiency as shown in Section \ref{sec:rates}.
See Algorithm 2.
\begin{figure}
\begin{center}
\fbox{\parbox{6in}{
\begin{center}
{\sf Algorithm 2:
Local Sandwich Slicer Algorithm}
\end{center}
\begin{enum}
\item Divide ${\cal X}$ into bins
$A_1,\ldots, A_m$.
\item Apply Algorithm 1
separately on all $Y_i$'s within each $A_k$.
\item Output $C_n^+(x)$: the resulting set
of $A_k$ for all $x\in A_k$.
\end{enum}
}}
\end{center}
\end{figure}

\begin{remark}\label{rem:flexible}
In the approach described above, the local
conformity measure is $\hat p^{(x,y)}(v|A_k)$.  In principle
one can use any conformity measure that
does not need to depend on the partition $A_k$, as long as the symmetry condition
is satisfied.
For example, one can use either the estimated joint density
$\hat p^{(x,y)}(u,v)$ or the estimated conditional density
$\hat p^{(x,y)}(v|u)$.  We note that when ${\sf diam}(A_k)$
is small, these choices of conformity measure are
close to each other since $p_X(x)$ and $p(\cdot|x)$ change
very little when $x$ varies inside $A_k$.
\end{remark}

\begin{remark} Although one can choose any conformity measure,
in order to have local validity  the
ranking must be based on a local subset of the sample.  When $A_k$
is small and the distribution is smooth enough,
the local sample $(X_{i_\ell}:1\le \ell\le n_k)$ approximates independent observations from
$p(\cdot|X=x)$ for $x\in A_k$, which can be used to approximate
the conditional oracle $C_P(x)$.
\end{remark}

\section{Asymptotic Properties}\label{sec:rates}

In this section we investigate the asymptotic efficiency of
the locally valid prediction band given in (\ref{eq:loc_val_band}).
The efficiency argument is similar for other choices of
conformity measures, such as joint density or conditional density.
Again, we focus on cases
where ${\rm supp}(P_X)=[0,1]^d$ and
$\mathcal A$ is a cubic histogram with width $w_n$.
The conformity measure is
$\hat p^{(x,y)}(Y_i|A_k)$ for $x\in A_k$,
where $\hat p^{(x,y)}(v|A_k)$ is defined as in equation
(\ref{eq:local_kernel_density_aug})
with kernel bandwidth  $h_n$.

\subsection{Notation}

In the subsequent arguments,
$p_X(\cdot)$ denotes the marginal density of $X$,
$p(y|x)$ the conditional density of $Y$ given $X=x$, and
$p(y|A_k)$ the conditional density of $Y$ given $X\in A_k$.
The kernel estimator of $p(y|A_k)$ is
denoted by $\hat p(\cdot|A_k)$ and
$\hat P(\cdot|A_k)$ is the empirical distribution
of $(Y|X\in A_k)$.

The upper and lower level sets of
conditional density $p(y|x)$ are denoted by
$L_x(t)\equiv \{y:p(y|x)\ge t\}$ and $L_x^\ell(t)
\equiv \{y:p(y|x)\le t\}$, respectively;
$\hat L_k(t)$, $\hat L_k^\ell(t)$ are
the counterparts of $L_x(t)$ and $L_x^\ell(t)$,
defined for $\hat p(\cdot|A_k)$.
As in the definition of
conditional oracle, $t_x^{(\alpha)}$ is solution to the
equation $P_x(L_x(t))=1-\alpha$.
Its existence and uniqueness is guaranteed if the contour
$\{y:p(y|x)=t\}$ has zero measure for all
$t>0$.
Finally we let $G_x(t)=P_x(L_x^\ell(t))$.

\subsection{The Sandwich Lemma}

Heuristically, $\hat p(y|A_k)\approx p(y|x)$ for
$x\in A_k$ when ${\rm diam}(A_k)$ is small
and $p(y|x)$ varies smoothly in $x$.  As a result,
the estimated densities $\hat p^{(x,y)}(Y_i|A_k)$
can be viewed as roughly a sample from $p(Y|x)$, and hence
$\hat C(x)$ approximates the conditional oracle
$C_P(x)$.
First we show that
$\hat C(x)$ can be approximated by
two plug-in conditional density level sets (Lemma \ref{lem:sandwich}).
For a fixed $A_k\in \mathcal A$,
conditioning on $(i_1,...,i_{n_k})$,
let $(X_{(k,\alpha)},Y_{(k,\alpha)})$
be the element of $\{(X_{i_1},Y_{i_1}),$
$...,(X_{i_{n_k}},Y_{i_{n_k}})\}$ such that
$\hat p(Y_{(k,\alpha)}|A_k)$ ranks
$\lfloor n_k\alpha\rfloor$ in ascending order
among all
$\hat p(Y_{i_{j}}|A_k)$, $1\le j\le n_k$.

\begin{lemma}[The Sandwich Lemma \citep{LeiRW11}]
\label{lem:sandwich}
For any fixed $\alpha\in(0,1)$, if $\hat
C(x)$ is defined in (\ref{eq:loc_val_band}) and
$||K||_\infty=K(0)$, then $\hat C(x)$
is ``sandwiched'' by two plug-in conditional density level sets:
\begin{equation}
\hat L\left(\hat
p\left(X_{(k,\alpha)},Y_{(k,\alpha)}|A_k\right)\right)
\subseteq \hat C(x)
\subseteq
\hat L\left(\hat p(X_{(k,\alpha)},Y_{(k,\alpha)}|A_k)-(n_kh_n)^{-1}\psi_K\right),
\end{equation}
where $\psi_K=\sup_{x,x'}|K(x)-K(x')|$.
\end{lemma}

The Sandwich Lemma
provides simple and accurate characterization of $\hat C(x)$
in terms of plug-in conditional density level sets, which are
much easier to estimate. The asymptotic properties of $\hat C(x)$ can be
obtained by those of the sandwiching sets.

\subsection{Rates of convergence}

To show the asymptotic efficiency of $\hat C(x)$, it
suffices to show efficiency for both sandwiching sets in Lemma
\ref{lem:sandwich}.  We need regularity conditions to quantify
and control the approximations $p(y|x)\approx p(y|A_k)$,
$\hat p(y|A_k)\approx p(y|A_k)$, and $\hat L_k (t)\approx L_x(t)$.

The following assumption puts boundedness and smoothness
conditions on the marginal density $p_X$, conditional
density $p(y|x)$, and its derivatives.

\noindent\textbf{Assumption A1 (regularity of marginal
and conditional densities)}
\begin{enum}
  \item [(a)]
The marginal density of $X$
   satisfies
   $0<p_0\le p_X(x)\le
   p_1<\infty$ for all $x$.
  \item [(b)] For all $x$,
  $p(\cdot|x)$ is H\"{o}lder class $\mathcal P(\beta, L)$. Correspondingly,
  the kernel $K$ is a valid kernel of order $\beta$.
  \item [(c)] For any $0\le s\le\lfloor \beta\rfloor$,
  $p^{(s)}(y|x)$ is continuous and uniformly bounded by $L$ for all $x,y$.
  \item [(d)] The conditional density is Lipschitz in
  $x$: $||p(\cdot|x)-p(\cdot|x')||_\infty \le L ||x-x'||$.
\end{enum}
The H\"{o}lder class of smooth functions and
valid kernels are common concepts in
nonparametric density estimation.  We give their definitions in
Appendix \ref{app:def}.
Assumptions A1(b) and A1(c) implies that $p(\cdot|A_k)$ is also
in a H\"{o}lder class and can be estimated well by kernel estimators.
A2(d) enables us to approximate $p(\cdot|x)$ by $p(\cdot|A_k)$ for
all $x\in A_k$.

The next assumption gives sufficient regularity condition
on the level sets $L_x(t)$.

\noindent\textbf{Assumption A2 (regularity of conditional density level set)}
\begin{enum}
  \item [(a)]  There exist
   positive constants
   $\epsilon_0$, $\gamma$, $c_1$, $c_2$,
    such that
  $$c_1(t_2-t_1)^\gamma \le G_x(t_2)-G_x(t_1)\le c_2(t_2-t_1)^\gamma,$$
  for all $t_x^{(\alpha)}-\epsilon_0\le t_1\le t_2\le
  t_x^{(\alpha)}+\epsilon_0$.
  \item [(b)]There exist positive constants $t_0$ and $C$, such that
  $0<t_0<\inf_xt_x^{(\alpha)}$ and $\mu(L_x(t_0))<C$ for all $x$.
  \end{enum}
Assumption A2(a) is related to the notion of ``$\gamma$-exponent''
condition introduced by \cite{Polonik95} and widely used
in the density level set literature \citep{Tsybakov97,RigolletV09}.
It ensures that the conditional density function
$p(\cdot|x)$ is neither too flat nor too steep near the contour
at level $t_x^{(\alpha)}$, so that the cut-off value $t_x^{(\alpha)}$
and the conditional density level set $C_P(x)$ can be approximated
from a finite sample.
As mentioned in \cite{AudibertT07}, if Assumption
A1(b) also holds, the oracle band $C_P(x)$
is non-empty only if $\gamma(\beta \wedge 1)\le 1$,
which holds for the most common case $\gamma=1$.
Part (b) simply simply puts some constraints on
 the optimal levels as well as the size
of the level sets.

The following critical rate will be used
repeatedly in our analysis.
\begin{equation}\label{eq:critical_rate}
r_n=\left(\frac{\log n}{n}\right)^{\frac{\beta}{\beta(d+2)+1}}.
\end{equation}

The rate may appear to be non-standard.
This is because we are assuming difference amounts of smoothness on $y$ and $x$.
This seems to be necessary to achieve both marginal and local validity.
We do not know of any procedure that uses a smoother construction and still retains
finite sample validity.
The next theorem gives
the convergence rate on the asymptotic efficiency of
the locally valid prediction band constructed in
Subsection \ref{subsec:local_valid_construct}.

\begin{theorem}\label{thm:main_consist}
Let $\hat C$ be the prediction band given by
the local conformity procedure as described in
(\ref{eq:loc_val_band}).  Choose
$w_n\asymp r_n$,
$h_n\asymp r_n^{1/\beta}$.
Under Assumptions A1-A2, for any $\lambda>0$, there exists
constant $A_\lambda$, such that
$$
\mathbb P\left(\sup_{x\in\mathcal X}\mu\left(\hat C(x)\triangle C_{\rm P}(x)\right)\ge A_\lambda
r_n^{\gamma_1}\right)=O(n^{-\lambda}),
$$
where $\gamma_1=\min(1,\gamma)$.
\end{theorem}

Thus, in the common case $\gamma=1$, the rate is $r_n$.
The following lemma follows easily from the previous result.

\begin{lemma}
Under assumptions A1 and A2,
the local band is asymptotically conditionally valid.
\end{lemma}

\begin{remark}
  It follows from the proof that the output of Algorithm
  2 also satisfies the same asymptotic efficiency and
  conditional validity results.
\end{remark}
\subsection{Minimax Bound}

The next theorem says that in the most common case $\gamma=1$, the rate
given in Theorem \ref{thm:main_consist} is indeed minimax rate
optimal.
We define the minimax risk by
\begin{equation}
\inf_{\hat C \in {\cal C}_{n,\alpha}}\sup_{P\in\mathcal P(\beta,L)}
\mathbb  E_P\mu\left[ \hat C(x)\triangle C(x)\right]
\end{equation}
where ${\cal C}_{n,\alpha}$
is the set of all valid prediction sets, and
$\mathcal P(\beta,L)$ is the class of distributions
satisfying A1 and A2 with $\gamma=1$.
We can obtain a lower bound on the
minimax risk by taking the infimum over all
set estimators $\hat C$, as in the following result.

\begin{theorem}[Lower bound on estimation error]
\label{thm:lowerbound}
Let $\mathcal P(\beta,L)$
be the class of distributions on
$[0,1]^d\times \mathbb R^1$ such that for each
$P\in\mathcal P(\beta,L)$,
$P_X$ is uniform on $[0,1]^d$, and satisfies
Assumptions A1-A2 with $\gamma=1$.
Fix an $\alpha\in(0,1)$,
there exist constant $c=c(\alpha,\beta, L,d)>0$ such that
$$
\inf_{\hat C}\sup_{P\in\mathcal P(\beta,L)}
\mathbb  E_P\mu\left[ \hat C(x)\triangle C(x)
\right]\ge c r_n.
$$
\end{theorem}

Hence, our procedure achieves the same rate as the
lower bound and so is minimax rate optimal
over the class
$\mathcal P(\beta,L)$.
The proof of
Theorem \ref{thm:lowerbound}
is in Section \ref{subsec::lower-proof}
and uses a somewhat non-standard construction.

\section{Tuning Parameter Selection}\label{sec:bandwidth}

In the band given by (\ref{eq:loc_val_band}),
there are two bandwidths to choose: $w_n$ and $h_n$.  Note that
since each bin $A_k$ can use a different $h_n$ to estimate the
local marginal density $\hat p(\cdot|A_k)$, we can consider
$h_{n,k}$, allowing a different kernel bandwidth for each bin.

Since all bandwidths give local validity, one can choose the
combination
of $(w_n,h_{n,k})$ such that the resulting conformal set
has smallest Lesbesgue measure.
 Such a two-stage
procedure
of selecting $w_n$ and $h_{n,k}$
from discrete candidate sets $\mathcal W=\{w^1,...,w^m\}$
and $\mathcal H=\{h^1,...,h^\ell\}$ is detailed in Algorithm 3.
To preserve finite sample marginal validity with data-driven bandwidths,
we split the sample into two equal-sized subsamples, and
apply the tuning algorithm on one subsample and use
the output bandwidth on the other subsample to obtain
the prediction band.

\begin{figure}
\begin{center}
\fbox{\parbox{6.3 in}{
\begin{center}{\sf Algorithm 3: Bandwidth Tuning for COPS}\end{center}
Input: Data $\mathcal Z$, level $\alpha$, candidate sets
 $\mathcal W$, $\mathcal H$.
\begin{enum}
\item Split data set into two equal sized subsamples,
$\mathcal Z_1$, $\mathcal Z_2$.
\item For each $w\in\mathcal W$
\begin{enum}
\item Construct partition $\mathcal A^w$.
\item For each $k$ and $h$ construct local conformal
prediction set
$\hat C_{h,k}^1$,
each at level $1-\alpha$, using data $\mathcal Z_1$.
\item Let $h^*_{w, k}=\arg\min_{h\in \mathcal H}
\mu\big(\hat C_{h,k}^1\big)$, for all $k$.
\item Let $Q(w)=\frac{1}{n}\sum_k n_k \mu\big(\hat C^1_{h^*_{w,
    k},k}\big)$.
\end{enum}
\item Choose $\hat w = \arg\min Q(w)$;
$\hat h_{\hat
    w,k}=h^*_{\hat w, k}$.
\item Construct partition $\mathcal A^{\hat w}$.
For $x\in A_k$, output prediction band $\hat C(x)=
\hat C_{\hat h_{\hat w,k},k}^2$, where
$\hat C_{h,k}^2$ is the local conformal prediction set
estimated from data $\mathcal Z_2$ in local set $A_k$.
\end{enum}
}}
\end{center}
\end{figure}

Following Remark \ref{rem:flexible}, one
can use different
conformity measures to construct $\hat C$.
In principle, the above
sample splitting procedure works for any
conformity measures.

It is straightforward to show that
the band $\hat C$ constructed as above
using data-driven tuning parameters
is locally valid and marginally valid, because
the bandwidth $(w,h)$ used is independent of
the training data $\mathcal Z_2$.
From the construction of $\hat C$,
it will have small excess risk
if the conformal prediction set is
stable under random sampling.  Then asymptotic
efficiency follows if one can relate the excess
risk to the symmetric difference risk.
A rigorous argument is beyond the scope of this
paper and will be pursued in a separate paper.

\section{Data Examples}
\label{sec::examples}

In this section we apply our method to
some examples.

\subsection{A Synthetic Example}

The procedure is illustrated by the following example in which
$d=1$, and
 \begin{equation}\label{eq:illustration_model}
\begin{array}{rl}
X\sim& {\sf Unif}[-1.5,1.5]\,,\\
(Y|X=x)\sim& 0.5N\left[f(x)-g(x),\sigma^2(x)\right]+
0.5N\left[f(x)+g(x),\sigma^2(x)\right]\,,
\end{array}\end{equation}
where
\begin{align*}
  f(x)=&(x-1)^2(x+1),\\
  g(x)=&2\sqrt{x+0.5}\times\ind(x\ge -0.5),\\
  \sigma^2(x)=&1/4+|x|.
\end{align*}
  For $x\le -0.5$, $(Y|X=x)$ is a
Gaussian centered at $f(x)$
with varying variance $\sigma^2(x)$. For
$x\ge -0.5$, $(Y|X=x)$ is a two-component Gaussian mixture,
and for large values of $x$, the two components have little
overlap.

\begin{figure}
\begin{center}
\includegraphics[scale=0.9]{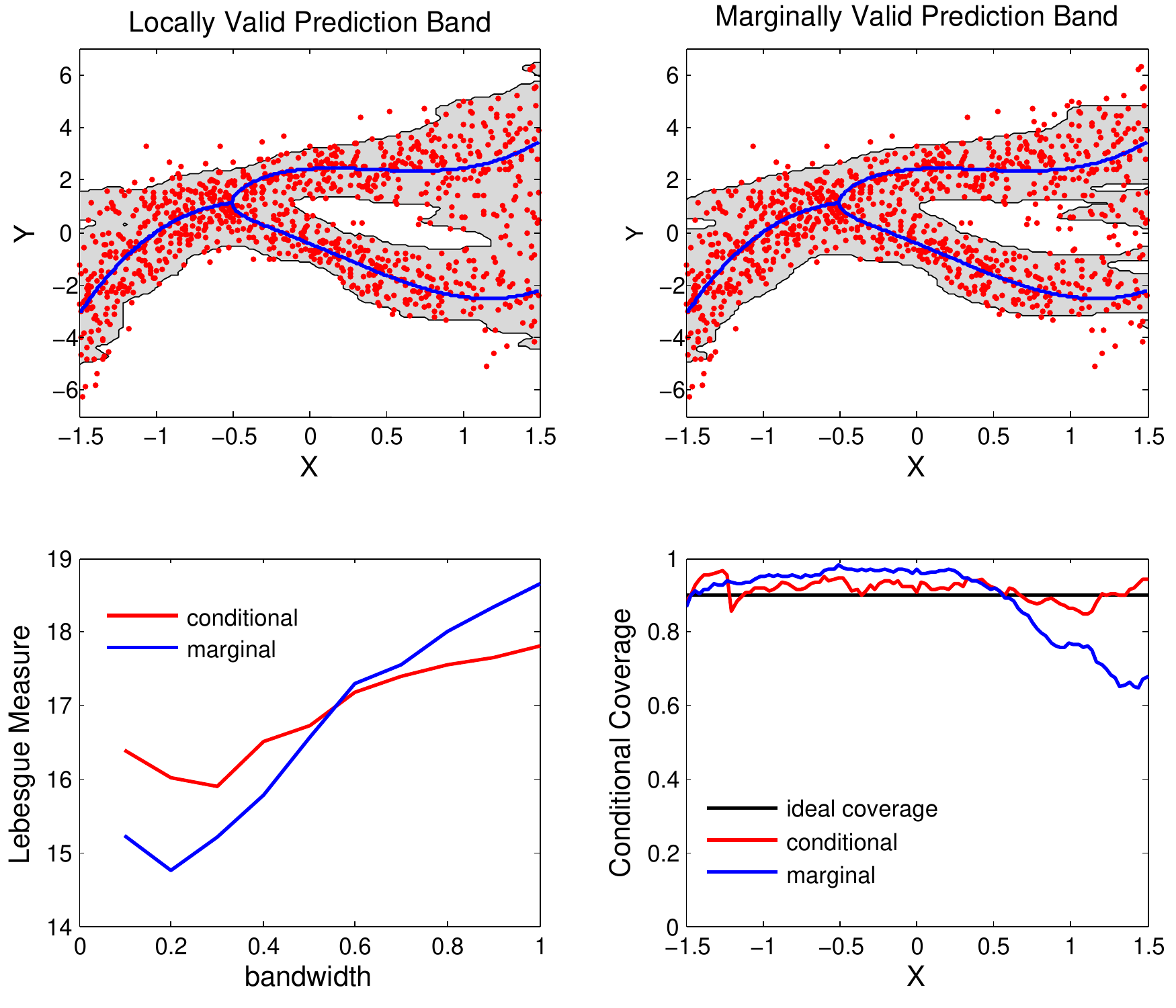}
\end{center}
\caption{Conditional and marginal prediction bands.
The bottom left panel shows the relationship between bandwidth and
Lebesgue measure of the prediction band. The bottom right panel
shows the conditional coverage of the estimated set $\hat C(x)$
as a function of $x$.}
\label{fig:cond_vs_marg}
\end{figure}

The performance of prediction bands using local conformity
is plotted
and compared with the marginal valid band in Figure \ref{fig:cond_vs_marg},
with $n=1000$, $\alpha=0.1$.
The conformity measure used here is $\hat p^{(x,y)}(Y_i|X_i)$.
The locally valid prediction band is constructed by partitioning
the support of $P_X$ into 10 equal sized bins, whereas the
marginally valid band is constructed by a global ranking with
the same conformity measure.
We see that although the locally valid band has larger Lebesgue measure,
it gives the desired coverage for all values $x$.
The marginally valid band over covers for smaller values of $x$,
and under covers for larger values of $x$.
We also plot the effect of bandwidth on the
size of prediction set (lower left panel of Figure \ref{fig:cond_vs_marg}).

\subsection{Car Data}

Next we consider an example on car mileage.
The original data contains features for
about 400 cars.  For each car, the data
consist of miles per gallon, horse power,
engine displacement, size, acceleration,
number of cylinders, model year, origin of
manufacture.  These data have been used in
statistics text books (for example,
\citet{DeGrootS12}, Chapter 11) to illustrate
the art of linear regression analysis.  Here
we reproduce the linear model built in
Example 11.3.2 of \cite{DeGrootS12},
where we want to predict the miles per gallon
by the horse power.  Clearly, the relationship
between miles per gallon and horse power is far from
linear (Figure \ref{fig:car}) so some transformation
must be applied prior to linear model fitting.
It makes sense to assume, both from intuition
and data plots,
that the inverse of miles per gallon, namely, gallons
per mile, has roughly a linear dependence on the
horse power.

In the right panel of Figure \ref{fig:car}
we plot the level 0.9 prediction band obtained from
the linear regression prediction band.  The
overall coverage is reasonably close to the nominal level.
However, due to the non-uniform noise level, the band
is too wide for small values of horse power and
too narrow for large values.
In the left panel, we plot the nonparametric conformal prediction
band using conformity measure $\hat p_h^{y}(Y_i|X_i)$ to
enhance smoothness of the estimated band.  Such a band is asymptotically
close to the one given in (\ref{eq:loc_val_band}).  The bandwidths are
$h_x=14$ and $h_y=1.4$.  The partition $\mathcal A$ is constructed by
partitioning the range of horse power into several intervals
to ensure each set $A_k$ contains roughly same number of sample points.
Here the tuning parameter is the number of partitions and is set to 8.

The advantage of our method is clear.  First, it automatically
outputs good prediction bands without involving choosing
the variable transformation.  The tuning parameters can be
chosen in either an automated procedure as described in
Algorithm 3, or by conventional choices (kernel bandwidth
selectors).  Second, the conformal prediction band is truly
distribution free, with valid coverage for all distribution and all
sample sizes.

\begin{figure}
\begin{center}
\includegraphics[scale=0.65]{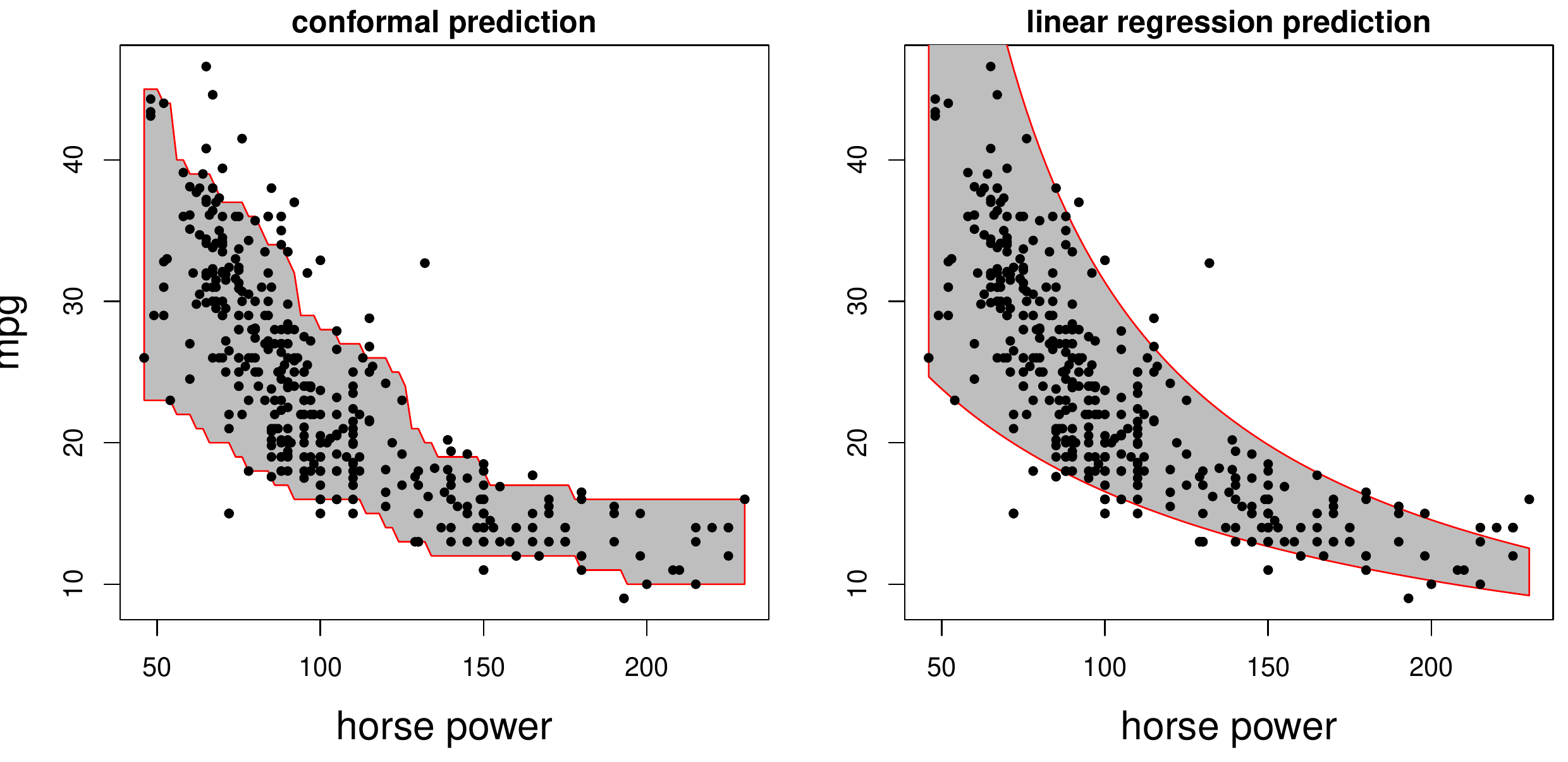}
\end{center}
\caption{Level $.9$ prediction bands using local conformal prediction (left)
and linear regression with variable transformation (right).}
\label{fig:car}
\end{figure}

\section{Final Remarks}
\label{sec::conclusion}

We have constructed nonparametric prediction bands
with finite sample, distribution free validity.
With regularity assumptions,
the band is efficient in the sense of achieving
the minimax bound.
The tuning parameters are completely data-driven.
We believe this is the first prediction band with these properties.

An important open question is to establish a rigorous
result on the asymptotic efficiency for the data-driven bandwidth.
A sketch of such an argument can be given by combining
two facts. First, the empirical average
excess loss $n^{-1}\sum_{k}n_k \mu(\hat C_{h,k})$
is a good approximation to the excess risk
$\E\left[\int\mu(\hat C_{h,k}(x))P_X(dx)\right]$
for all $w$
and $h$.  This problem is technically similar to
those considered by \cite{RinaldoSNW10} in
the study of stability of plug-in density level
sets and prediction sets.
Second, one can show that the excess risk
provides an upper bound of the symmetric difference
risk $\E (\hat C\triangle C_P)$, as
given in \cite{LeiRW11} (see also
\cite{ScottN06}).

The bands are not suitable for high-dimensional regression problems.
In current work, we are developing methods for constructing prediction bands
that exploit sparsity assumptions.
These will yield valid prediction and variable selection
simultaneously.

\section{Appendix}

In the appendix, we give supplementary technical details.

\subsection{Technical Definitions}\label{app:def}

Now we give formal definitions of some technical terms
used in the asymptotic analysis, including H\"{o}lder class
of functions and valid kernel functions of order $\beta$. These
definitions can be found in standard nonparametric inference
text books such as \cite[Section 1.2]{Tsybakov09}.
\begin{definition}[H\"{o}lder Class]
Given $L>0$, $\beta>0$.  Let $\ell=\lfloor\beta\rfloor$.
The H\"{o}lder class $\Sigma(\beta, L)$ is the family of functions
$f:\mathbb R^1\mapsto \mathbb R^1$ whose derivative $f^{(\ell)}$
satisfies
$$|f^{(\ell)}(x)-f^{(\ell)}(x')|\le L|x-x'|^{\beta-\ell},
~\forall~x,x'.$$
\end{definition}
\textbf{Remark:} If $f\in \Sigma(\beta,L)$, then
$f$ can be uniformly approximated by its local polynomials of
order $\ell$.  Define
$$f_{x_0}^{(\ell)}(x)=f(x_0)+f'(x_0)(x-x_0)+....+
\frac{f^{(\ell)}(x_0)}{\ell!}(x-x_0)^\ell.
$$
Then
$$|f(x)-f_{x_0}^{(\ell)}(x)|\le \frac{L}{\ell !}|x-x_0|^\beta.$$

\begin{definition}
  [Valid Kernels of Order $\beta$]
  Let $\beta>0$ and $\ell=\lfloor\beta\rfloor$. Say that $K:\mathbb R^1\mapsto
  \mathbb R^1$ is a valid kernel of order $\beta$ if the functions
  $u\mapsto u^jK(u)$, $j=1,...,\ell$, are integrable and satisfy
  $$\int K(u)du=1,~~\int u^jK(u)du=0,~j=1,...,\ell.$$
\end{definition}
\textbf{Remark:} The relationship between a H\"{o}lder
class $\Sigma(\beta,L)$ and a valid kernel $K$ of order
$\beta$ is that for any $p\in\Sigma(\beta,L)$, and $h=o(1)$,
$||p-p\star K_h||_\infty\le \frac{L}{\ell!}h^\beta$, where
$\star$ is the convolution operator and $K_h(x)=h^{-1}K(x/h)$.
\subsection{Proof of Lemma \ref{lem:impossible}}
\label{app:pf_impossible}
\begin{proof}[Proof of Lemma \ref{lem:impossible}]
For simplicity we prove the case where $d=1$.
Let
$$
{\sf TV}(P,Q) = \sup_A |P(A)-Q(A)|
$$
denote the total variation distance between $P$ and $Q$.
Given any $\epsilon>0$ define
$$
\epsilon_n = 2[1-(1-\epsilon^2/8)^{1/n}].
$$
From Lemma A.1 of \cite{Donoho1988},
if
${\sf TV}(P,Q) \leq \epsilon_n$ then
${\sf TV}(P^{n},Q^n) \leq \epsilon$.

Fix $\epsilon>0$.
Let $x_0$ be a non-atom
and choose $\delta$ be such that
$0 < P_X\left[B(x_0,\delta)\right] < \epsilon_n$
where $\epsilon_n = 2[1-(1-\epsilon^2/8)^{1/n}]$.
It follows that
${\sf TV}(P^{n},Q^{n}) \leq \epsilon$.

Fix $B>0$ and
let $B_0=B/(2(1-\alpha))$.
Define another distribution $Q$ by
$$
Q(A) = P(A\cap S^c) + U(A \cap S)
$$
where
$$
S = \Bigl\{ (x,y) :\ x\in B(x_0,\delta),\ y\in \mathbb{R}\Bigr\}
$$
and $U$ has total mass $P(S)$ and is uniform on
$\{ (x,y) :\ x\in B(x_0,\delta),\ |y| < B_0\}$.
Note that $P(S)>0$, $Q(S)>0$ and
${\sf TV}(P,Q) \leq \epsilon_n$.
It follows that
${\sf TV}(P^{n},Q^{n}) \leq \epsilon$.

Note that, for all
$x\in B(x_0,\delta)$,
$\int_{C(x)} q(y|x) dy \geq 1-\alpha$ implies that
$\mu[C(x)] \geq 2(1-\alpha)B_0 = B$.
Hence,
$$
Q^{n}\Biggl(\esssup_{x\in B(x_0,\delta)}\mu[C(x)]\ge B\Biggr) =1.
$$
Thus,
$$
P^{n}\Biggl(\esssup_{x\in B(x_0,\delta)}\mu[C(x)]\ge B\Biggr) \geq
Q^{n}\Biggl(\esssup_{x\in B(x_0,\delta)}\mu[C(x)]\ge B\Biggr)
-\epsilon =1 - \epsilon.
$$
Since $\epsilon$ and $B$ are arbitrary,
the result follows.
\end{proof}

\subsection{Proofs of asymptotic efficiency}
\begin{lemma}
  \label{lem:p_k_infty}
  Given $\lambda >0$, under condition A2
  and
  A4, there exists
  numerical constant $\xi_\lambda$ such
  that,
  $$\mathbb P\left(\sup_{k}
  ||\hat p(\cdot|A_k)
  -p(\cdot|A_k)||_\infty
  \ge \xi_\lambda r_n\right)=O(n^{-\lambda}).$$
\end{lemma}
\begin{proof}
for any fixed $k$,
$Y_{i_1},...,Y_{i_{n_k}}$ is a random sample from
$P(y|A_k)$ conditioning on $n_k$.

Let $\bar p(y|A_k)$ be the convolution
density $p(\cdot|A_k)\star K_{h_n}(\cdot)$, then using a result
from \citet{GineG02}, there
exists numerical constants
$C_1$, $C_2$ and $\xi_0$ such that for all $\xi\ge \xi_0$,
\begin{equation}
\label{eq:Gine}\mathbb P\left(||\hat p(\cdot|A_k)-\bar p(\cdot|A_k)||_\infty
\ge \xi\sqrt{\log n_k / (n_k
h_n)}\right)\le C_1h_n^{C_2\xi^2}.\end{equation}
On the other hand, by H\"{o}lder condition of $p(y|x)$ and hence on $p(\cdot|A_k)$,
we have:
$$||\bar p(\cdot|A_k)-p(\cdot|A_k)||_\infty \le Lh_n^\beta.$$
Put together with union bound on all $A_k\in \mathcal A_n$
$$\mathbb P \left(\exists k:~||\hat p(\cdot|A_k)-p
(\cdot|A_k)||_\infty \ge \xi\sqrt{
\log n_k/(n_kh_n)}+Lh_n^\beta
\right)\le C_1h_n^{C_2\xi^2}w_n^{-d}.$$
Consider event $E_0$:
$$E_{0}=\{b_1nw_n^d/2\le n_k\le 3b_2nw_n^d/2\},$$
where the constants $b_1$ and $b_2$ is defined as
in Assumption A1(a).
By Lemma \ref{lem:local_sample_size} we have
$$\mathbb P(E_0^c)\le C_3w_n^{-d}\exp(-C_4nw_n^d),$$
with constants $C_3$, $C_4$ defined in lemma
\ref{lem:local_sample_size}.

On $E_0$ and for $n$ large enough
we have
$$\sqrt{\frac{\log n_k}{n_kh_n}}\le 2
\sqrt{\frac{2\beta+1}{c_1\left[\beta(d+2)+1\right]}}
\sqrt{\frac{\log n}{nw_n^dh_n}}.$$

Note that
under Assumption A4, $\sqrt{\frac{\log n}{nw_n^dh_n}}=h_n^\beta=r_n$.

Let
$$\xi_\lambda=2
\sqrt{\frac{2\beta+1}{c_1\left[\beta(d+2)+1\right]}}
\left(\sqrt{\frac{\lambda(\beta(d+2)+1)+\beta
d}{C_2}}\bigvee \xi_0\right)+L,$$
where the constant $c_1$ is defined in Assumption A1(a),
$C_2$ defined in equation (\ref{eq:Gine}), and $L$
defined in A2(a).

Then we have
\begin{align*}
&\mathbb P\left(\sup_k ||\hat p(\cdot|A_k)
-p(\cdot|A_k)||_\infty
\le \xi_\lambda r_n\right)\\
\ge&\mathbb P\left(\sup_k||\hat p(\cdot|A_k)
-p(\cdot|A_k)||_\infty\le
(\xi_\lambda-L)\sqrt{\frac{\log
n}{nw_n^dh_n}}+Lh_n^\beta,~E_0
\right)\\
\ge&\mathbb P\left(\sup_k
||\hat p(\cdot|A_k)-p(\cdot|A_k)||_\infty
\le \frac{\xi_\lambda-L}
{2\sqrt{\frac{2\beta+1}{c_1\left[\beta(d+2)+1\right]}}}
\sqrt{\frac{\log n_k}{n_kh_n}}+Lh_n^\beta,~E_0
\right)\\
\ge & 1-\mathbb P\left(\exists k:||\hat
p(\cdot|A_k)-p(\cdot|A_k)||_\infty\ge
\frac{\xi_\lambda-L}
{2\sqrt{\frac{2\beta+1}{c_1\left[\beta(d+2)+1\right]}}}
\sqrt{\frac{\log n_k}{n_kh_n}}+Lh_n^\beta\right)-\mathbb
P(E_0^c)\\
=&1-O(n^{-\lambda}),\end{align*}

\end{proof}

  \begin{corollary}\label{cor:R_n}
  Let $R_n(x)=||\hat p_n(y|A_k)-p(y|x)||_\infty$, then for
  any $\lambda >0 $, there exists $\xi_{1,\lambda} >0$ such that
  $$ \mathbb P\left[\sup_{x\in B_k}R_n(x)\ge
      \xi_{1,\lambda} r_n\right]=O(n^{-\lambda}).$$
  \end{corollary}
  \begin{proof}
First by Lipschitz condition A2(c) on $p(y|x)$,
$$||p(y|A_k)-p(y|x)||_\infty\le \sqrt{d}Lw_n.$$

Note that $w_n=r_n$, the claim then follows by applying
Lemma \ref{lem:p_k_infty} and choose
$\xi_{1,\lambda}=\xi_\lambda+\sqrt{d}L$.
\end{proof}

\begin{lemma}\label{lem:V_n} Let
  $$V_n(x)=\sup_{t\ge t_0}\left|\hat
  P(L^\ell_x(t)|A_k)-P(L^\ell_x(t)|x)\right|,$$
then, for any $\lambda >0$, there exists $\xi_{2,\lambda}$ such that
  $$\mathbb P\left(\sup_{x\in \mathcal X}V_n(x)\ge
  \xi_{2,\lambda}r_n^{\gamma_1}\right)=O(n^{-\lambda}),$$
  with $\gamma_1=\min(\gamma, 1)$.
\end{lemma}
\begin{proof} Consider a fixed
 $A_k$ and an $x\in A_k$.
Note that $\{L^\ell_x(t):t\ge t_0\}$ is a nested
     class of sets with VC dimension 2. By classical empirical process
     theory, for all $B>0$ we have
     \begin{equation}\label{eq:VC}
     \mathbb P\left(\sup_t\left|\hat P(L^\ell_x(t)|A_k)-P(L^\ell_x(t)|A_k)\right|\ge
     B\sqrt{\frac{\log n_k}{
        n_k}}\right)\le C_0 n_k^{-(B^2/32-2)},\end{equation}
        for some universal constant $C_0$.

     On the other hand
    \begin{align}\label{eq:V_n_local_lip}
      \left|P(L_x^\ell(t)|A_k)-P(L_x^\ell(t)|x)\right|
      &=\left|\int_{L^\ell_x(t)}(p(y|A_k)-p(y|x))dy\right|
      \nonumber\\
&\le\sqrt{d}Lw_n\mu(L_x(t))\nonumber\\
&\le \sqrt{d}Lw_n\mu(L_x(t_0))\nonumber\\
&\le CL\sqrt{d} w_n,
    \end{align}
where the constant $C$ is defined in Assumption
A3(b).

Note that on $E_0$ we have
$\sqrt{\log n_k/ n_k}=o(r_n)$ and hence
$\sqrt{\log n_k / n_k}\le r_n$ for $n$
large enough.

Consider any $x'\in A_k$.
\begin{align}\label{eq:local_G_lip}
    &\left|\hat P(L_{x'}^\ell(t)|A_k)-
    P(L_{x'}^\ell(t)|x')\right|\nonumber\\
    \le&
    \left|\hat P(L_{x'}^\ell(t)|A_k)-\hat P
    (L_x^\ell(t)|A_k)\right|+
    \left|\hat P(L_x^\ell(t)|A_k)-P(L_x^\ell(t)|x)\right|
    + \left|P(L_x^\ell(t)|x)-P(L_{x'}^\ell(t)|x')\right|
    \nonumber\\
    \le&  ||\hat p(\cdot|A_k)||_\infty\mu\left(
    L_x^\ell(t)\triangle L_{x'}^\ell(t)\right)
    +V_n(x) + |G_x(t)-G_{x'}(t)|\nonumber\\
    \le &||\hat
    p(\cdot|A_k)||_\infty\frac{c_2(2L)^\gamma}{t_0}w_n^\gamma
    +V_n(x) + C\sqrt{d}Lw_n+\frac{c_22^\gamma
    L^{\gamma+1}}{t_0}w_n^\gamma,
\end{align}
where the last step uses Lemma \ref{lem:Lip_local_measure}
to control $\mu\left(L_x^\ell(t)\triangle
L_{x'}^\ell(t)\right)$ and $G_x(t)-G_{x'}(t)$.

Lemma 15 implies that, except for a probability of $O(n^{-\lambda})$, $\sup_k||\hat
p(\cdot|A_k)||_\infty=L+o(1)$ with $L$ defined in A2(b).
Combining (\ref{eq:VC}), (\ref{eq:V_n_local_lip}), and
(\ref{eq:local_G_lip}), we have, for some constant
$\xi_{2,\lambda}$
$$\mathbb P\left(\sup_x V_n(x)\ge
\xi_{2,\lambda}r_n^{\gamma_1}\right)=O(n^{-\lambda}),$$
where $\gamma_1=\min(\gamma,1)$.

%
%
%
%
%
%
\end{proof}

\begin{lemma}\label{lem:Lip_local_measure}
  Under assumptions A1-A3,
  $$\sup_k\sup_{t\ge t_0, x,x'\in
  A_k}|G_x(t)-G_{x'}(t)|=O(w_n^{\gamma \wedge 1}).$$
\end{lemma}
\begin{proof}
\begin{align}
  &L_x(t)\triangle L_{x'}(t)\nonumber\\
  =&
  \{y:p(y|x)>t,p(y|x')\le t\}\cup
  \{y:p(y|x)\le t, p(y|x')> t\}\nonumber\\
  =&\{y:t<p(y|x)\le t+Lw_n,p(y|x')\le t\}\cup
  \{y:t-Lw_n<p(y|x)\le t,p(y|x')>t\}\nonumber\\
  \subseteq &\{y:t-Lw_n< p(y|x)\le t+Lw_n\},
\end{align}
where the first step uses the fact that
$||p(\cdot|x)-p(\cdot|x')||_\infty\le L||x-x'||$
and the constant $L$ is from Assumption A2(c).
  \begin{align}
    &\left|G_x(t)-G_{x'}(t)\right|\nonumber\\
    \le&\left|P(L_x^\ell(t)|x)-P(L_x^\ell(t)|x')\right|
    +\left|P(L_x^\ell(t)|x')-P(L_{x'}^\ell(t)|x')\right|
    \nonumber\\
    =&\left|P(L_x(t)|x)-P(L_x(t)|x')\right|
    +\left|P(L_x^\ell(t)|x')-P(L_{x'}^\ell(t)|x')\right|
    \nonumber\\
    \le& \mu(L_x (t))||p(\cdot|x)-p(\cdot|x')||_\infty
    +||p(\cdot|x')||_\infty\mu(L_x^\ell(t)\triangle
     L_{x'}^\ell(t))\nonumber\\
     \le &
     C\sqrt{d}Lw_n+L\frac{G_{x'}(t+Lw_n)-G_{x'}(t-Lw_n)}{t_0}
     \nonumber\\
     \le &
  C\sqrt{d}Lw_n+\frac{c_22^\gamma
  L^{\gamma+1}}{t_0}w_n^\gamma,
  \end{align}
where the constant $L$ is from Assumption A2 and
$(c_2,C,\gamma)$ are defined in Assumption A3.
\end{proof}

We complete the argument using \cite{CadrePP09} and \cite{LeiRW11}.

\begin{lemma}
\label{lem:generic_argument}
Fix $\alpha>0$ and $t_0>0$.
Suppose $p$ is a density function satisfying Assumption
A3(a).
Let $\hat p$ be an estimated density such that $||\hat
p-p||_\infty \le \nu_1$,
and $\hat P$ be a probability measure satisfying
$\sup_{t\ge t_0}\left|\hat
P(L^\ell(t))-P(L^\ell(t))\right|\le \nu_2$.
Define
$\hat t^{(\alpha)}=\inf\{t\ge 0:\hat P(\hat L^\ell(t))\ge
\alpha\}$.
If $\nu_1,\nu_2$ are small enough such that
$\nu_1+c_1^{-1/\gamma}\nu_2^{1/\gamma}\le
t^{(\alpha)}-t_0$
and $c_1^{-1/\gamma}\nu_2^{1/\gamma}\le \epsilon_0$ (where $c_1$, $\gamma$ are constants in Assumption A3(a)),
then
 \begin{equation}\label{eq:bound_t_hat}
 \left|\hat t^{(\alpha)}-t^{(\alpha)}\right|\le \nu_1+c_1^{-1/\gamma}\nu_2^{1/\gamma}.
 \end{equation}

Moreover, for any $\tilde t^{(\alpha)}$ such that
$|\tilde t^{(\alpha)}-\hat t^{(\alpha)}|\le \nu_3$, if
$2\nu_1+c_1^{-1/\gamma}\nu_2^{1/\gamma}+\nu_3\le \epsilon_0$, then there exist
constants $\xi_1$, $\xi_2$ and $\xi_3$ such that
$$
\mu\left(\hat L(\tilde t^{(\alpha)})\triangle L(t^{(\alpha)})\right)\le
\xi_1\nu_1^{\gamma}+\xi_2\nu_2 + \xi_3\nu_3^\gamma.
$$
\end{lemma}
\begin{proof}
The proof follows essentially from \citet{LeiRW11},
which is a modified version of the argument
used in \citet{CadrePP09}.

For $t\ge t_0$, let $\hat L^\ell(t)=\{y:\hat p(y)\le t\}$. By the assumptions in the lemma we have
\begin{eqnarray*}
&&  L^\ell(t-\nu_1)\subseteq \hat L^\ell(t)\subseteq L^\ell(t+\nu_1)\\
&  \Rightarrow & \hat P(L^\ell(t-\nu_1))\le \hat P(\hat L^\ell(t))\le\hat   P(L^\ell(t+\nu_1))\\
&  \Rightarrow & P(L^\ell(t-\nu_1))-\nu_2\le \hat P(\hat L^\ell(t))   \le P(L^\ell(t+\nu_1))+\nu_2.
\end{eqnarray*}
Hence,
$$
\hat P(\hat   L^\ell(t^{(\alpha)}-\nu_1-c_1^{-1/\gamma}\nu_2^{1/\gamma}))
\le P(L^\ell(t^{(\alpha)}-c_1^{-1/\gamma}
\nu_2^{1/\gamma}))+\nu_2\nonumber \le  \alpha,
$$
where the last step uses Assumption A3(a).

Therefore, we must have $\hat t^{(\alpha)}\ge
t^{(\alpha)}-\nu_1-c_1^{-1/\gamma}\nu_2^{1/\gamma}$.  A
similar argument gives $\hat t^{(\alpha)}
\le t^{(\alpha)}+\nu_1+c_1^{-1/\gamma}\nu_2^{1/\gamma}.$
This proves the first part.

For the second part,
note that
$$\hat L(\tilde t^{(\alpha)})\triangle L(t^{(\alpha)})
=\{y:\hat p(y)\ge\tilde t^{(\alpha)}, p(y)<
t^{(\alpha)}\}\cup
\{y:\hat p(y)<\tilde t^{(\alpha)}, p(y)\ge
t^{(\alpha)}\}.
$$
By the assumption on $\tilde t^{(\alpha)}$
and the first result,
$$\{\hat p(y)\ge \tilde t^{(\alpha)}\}\subseteq
\{p(y)\ge
t^{(\alpha)}-2\nu_1-c_1^{-1/\gamma}\nu_2^{1/\gamma}
-\nu_3\},$$
$$\{\hat p(y)< \tilde t^{(\alpha)}\}\subseteq
\{p(y)<
t^{(\alpha)}+2\nu_1+c_1^{-1/\gamma}\nu_2^{1/\gamma}
+\nu_3\}.$$

As a result,
\begin{eqnarray*}
\mu\left(\hat L(\tilde t^{(\alpha)})\triangle  L(t^{(\alpha)})\right) & \leq &
\mu\left(\left\{y:|p(y)-t^{(\alpha)}|\le 2\nu_1+c_1^{-1/\gamma}\nu_2^{1/\gamma}+\nu_3\right\}\right)\\
&\le & t_0^{-1}c_2 (4\nu_1+2c_1^{-1/\gamma}\nu_2^{1/\gamma}+2\nu_3)^\gamma \leq
\xi_1\nu_1^\gamma +\xi_2\nu_2+\xi_3\nu_3^\gamma,
\end{eqnarray*}
where $(\xi_1,\xi_2,\xi_3)$ are functions of
$(t_0, c_1,c_2,\gamma)$.
\end{proof}

\begin{proof}
  [Proof of Theorem \ref{thm:main_consist}]
  The proof is based on a direct application of Lemma
  \ref{lem:generic_argument} to
  the density $p(\cdot|x)$ and the empirical measure
  $\hat P(\cdot|A_k)$ and
  estimated density function $\hat p(\cdot|A_k)$.

Here we use $\hat L$ for upper level sets of $\hat
p(\cdot|A_k)$ and omit the dependence on $k$.

Conditioning on $(i_1,...,i_{n_k})$,
then one can show that
the local conformal prediction set $\hat
C^{(\alpha)}(x)$
is ``sandwiched'' by two estimated level sets:
$$\hat L\left(\hat
p\left(X_{(i_\alpha)},Y_{(i_\alpha)}|A_k\right)\right)
\subseteq \hat C^{(\alpha)}(x)
\subseteq
\hat L\left(\hat p(X_{(i_\alpha)},Y_{(i_\alpha)}|A_k)-(n_kh_n)^{-1}\psi_K\right),$$
where $\psi_K=\sup_{x,x'}|K(x)-K(x')|$.
So the asymptotic properties of $\hat C^{(\alpha)}(x)$ can be
obtained by those of the sandwiching sets.

Recall that $(X_{(i_\alpha)},Y_{(i_\alpha)})$
is the element of $\{(X_{i_1},Y_{i_1}),...,
(X_{i_{n_k}},Y_{i_{n_k}})\}$ such that
$\hat p(Y_{(i_\alpha)}|A_k)$ ranks
$\lfloor n_k\alpha\rfloor$ in ascending order
among all
$\hat p(Y_{i_{j}}|A_k)$, $1\le j\le n_k$.
 Let $\hat t^{(\alpha)}= \hat
 p(X_{(i_\alpha)},Y_{(i_\alpha)})$.  It is easy to check
that
$$\hat t^{(\alpha)}=\inf\left\{t\ge 0:\hat P\left(\hat L^\ell(t)\big|A_k\right)
\ge
\alpha\right\}.$$
Consider event
$$E_1=\left\{\sup_x R_n(x)\le \xi_{1,\lambda}r_n,~\sup_x
V_n(x)\le \xi_{2,\lambda}r_n^{\gamma_1}\right\},$$
where $\xi_1$ and $\xi_2$ are defined as
in the statement of Corollary \ref{cor:R_n}
and Lemma \ref{lem:V_n}. We have
$\mathbb P(E_1^c)=O(n^{-\lambda})$.

Let $\nu_1=\xi_{1,\lambda}r_n$, $\nu_2=\xi_{2,\lambda}
r_n^{\gamma_1}$. Note that $r_n\rightarrow 0$ as
$n\rightarrow \infty$,
so for $n$
large enough, we have $\nu_1$ and $\nu_2$ satisfying the
requirements in
Lemma \ref{lem:generic_argument}.  Let $\nu_3=0$ in this
case,
then we have, for some constants $\xi_{1,\lambda}'$,
$\xi_{2,\lambda}'$, that
$$\mathbb P\left(\sup_x\mu\left(\hat L(\hat t^{(\alpha)})\triangle
L_x(t^{(\alpha)})\right)\ge \xi_{1,\lambda}' r_n^{\gamma}+\xi_{2,\lambda}'
 r_n^{\gamma_1}\right)=O(n^{-\lambda}),$$
which is equivalent to
$$\mathbb P\left(\sup_x\mu\left(\hat L(\hat t^{(\alpha)})\triangle
L_x(t^{(\alpha)})\right)\ge \xi_{\lambda}'
 r_n^{\gamma_1}\right)=O(n^{-\lambda}),$$
for some constant $\xi_\lambda'$ independent of
$n$.

Now let $\tilde t^{(\alpha)}=\hat t^{(\alpha)}-(n_kh_n)^{-1}\psi_K$.  Applying
Lemma \ref{lem:generic_argument} with $\nu_{3}=\nu_{3,n}=(n_kh_n)^{-1}\psi_K$, we obtain,
for some constants $\xi_{j,\lambda}''$, $j=1,2,3$,
$$\mathbb P\left(\mu\left(\hat L(\hat t^{(\alpha)})\triangle
L_x(t^{(\alpha)})\right)\ge \xi_{1,\lambda}'' r_n^{\gamma}+\xi_{2,\lambda}''
 r_n^{\gamma_1} +
\xi_{3,\lambda}'' \nu_{3,n}^\gamma\right)=O(n^{-\lambda}).$$
Note that on $E_0$, $\nu_{3,n}=o(r_n)$, so the above inequality reduces to
$$\mathbb P\left(\mu\left(\hat L(\hat t^{(\alpha)})\triangle
L_x(t^{(\alpha)})\right)\ge \xi_{\lambda}''
 r_n^{\gamma_1}\right)=O(n^{-\lambda}).$$

The conclusion of Theorem \ref{thm:main_consist} follows from the sandwiching property:
$$\mu\left(\hat C^{(\alpha)}(x)\triangle L_x(t_x^{(\alpha)})\right)
\le
\mu\left(\hat L\left(\hat t^{(\alpha)}\right)\triangle
L_x(t_x^{(\alpha)})\right)
+\mu\left(\hat L\left(\tilde t^{(\alpha)}\right)\triangle
L_x(t_x^{(\alpha)})\right),$$
where $\hat t^{(\alpha)}= \hat p(X_{(i_\alpha)},Y_{(i_\alpha)})$ and
$\tilde t^{(\alpha)}=\hat t^{(\alpha)}-(n_kh_n)^{-1}\psi_K$.
\end{proof}

\begin{lemma}
  [Lower bound on local sample
  size]\label{lem:local_sample_size}
Under assumption A1:
  $$\mathbb P\left(\forall k: b_1nw_n^d/2\le n_k\le 3b_2nw_n^d/2\right)
  \ge1- C_3w_n^{-d}e^{-C_4nw_n^d},$$
where $C_3=2\left[{\sf Diam}({\rm supp}(P_X))\right]^d$
and $C_4=b_1^2/(8b_2+4b_1/3)$ with $b_1$, $b_2$ defined
in Assumption A1(a).
\end{lemma}
\begin{proof}
 Let $p_k=P_X(A_k)$. Use Bernstein's inequality,
  for each $k$,
  \begin{align*}
    \mathbb P\left(|n_k-np_k|\ge t\right)\le \exp\left(
  -\frac{t^2/2}{np_k(1-p_k)+t/3}\right).
  \end{align*}
  The result follows by taking $t=c_1nw_n^d/2$
   and union bound.
\end{proof}

\subsection{Proof of Theorem \ref{thm:lowerbound}}
\label{subsec::lower-proof}
In the following proof we focus on
the rate and ignore the tuning on
constants.
The proof uses Generalized Fano's Lemma and
the construction
follows these key steps.
\begin{enumerate}
\item Let the marginal of $X$ be uniform on $[0,1]^d$.
Divide $[0,1]^d$ into cubes of size $w>0$.
\item Choose a density function $p_0(y)$ such that:
\begin{enumerate}
\item $p_0(y)$ is symmetric and H\"{o}lder
smooth of order $\beta$.
\item There exists $y_0<0$ and
$\delta>0$, such that $p_0'(y)=1$ for all
$y\in (y_0-\delta,y_0+\delta)$.
\end{enumerate}
\item For $x\in A_j$, let $x_j$ be
the center of $A_j$. Define conditional density:
$$
p(y|x)=p(y,x-x_j)=p_0(y)+h(x-x_j)
K\left(\frac{y-y_0}{h^{\frac{1}{\beta}}(x-x_j)}\right)- h(x-x_j)
K\left(\frac{y+y_0}{h^{\frac{1}{\beta}}(x-x_j)}\right),
$$
where $h(x)$ is a function defined on $R^d$ with
support on $[-w/2,w/2]^d$,
attaining its maximum at $0$, and satisfying
$||h'||_\infty\le M<\infty$,
$h'(x)=0$ for $||x||_\infty\ge w/2$.
In particular, take
$h(x)=w\eta(2x/w),$ where $\eta(x)$
is a $d$-dimensional kernel function supported on
$[-1,1]^d$.
It is easy to verify
that the following conditions hold:
\begin{enumerate}
\item $p(\cdot|x)$ is a density function for all $x$.
\item
$p(y|x)$ is H\"{o}lder smooth of order $\beta$.
This can be verified
by noting that both $p_0$  and
$h(x-x_j)
K\left(\frac{y-y_0}
{h^{\frac{1}{\beta}}(x-x_j)}\right)$ are
H\"{o}lder
smooth of order $\beta$.
\item $|p(y|x)-p(y|x')|\le L ||x-x'||$.
This can be verified by noting
that
\begin{align*}
&\left|\frac{\partial }{\partial x}p(y|x)\right|\\
=& \left|h'(x)K\left(\frac{y-y_0}{h^{1/\beta}(x-x_j)}\right)
-h(x)K'\left(\frac{y-y_0}{h^{1/\beta}(x-x_j)}\right)
\frac{1}{\beta}h^{-\frac{1}{\beta}-1}(x-x_j)(y-y_0)\right|
\\
\le& ||h'||_\infty||K||_\infty+||h'||_\infty||K'||_\infty \left|(y-y_0) h^{-\frac{1}{\beta}}(x-x_j)\right|\\
\le&||h'||_\infty||K||_\infty+||h'||_\infty||K'||_\infty
\end{align*}
\end{enumerate}
\item For $j=1,...,w^{-d}$, let $P_j$ be the distribution of $(X,Y)$
such that
\begin{enumerate}
\item $P_X$ is uniform.
\item $p_j(y|X=x)=p_0(y)$ for $x\notin A_j$.
\item $p_j(y|X=x)=p(y|x)$ for $x\in A_j$.
 \end{enumerate}
We can verify that the Lipschitz condition
$|p(y|x)-p(y|x')|\le L ||x-x'||$
still
holds if we require $h'=0$ on the border of
the histogram cube.
\item (Pairwise separation) For $i\neq j$,
The conditional density level sets at $p_0(y_0)$ differ at least
$ch$ for some constant $c$ (Consider $p_j(y|X=x_j)$ and $p_i(y|X=x_j)$ and note that
they corresponds to the same level $\alpha$
as prediction bands).
\item (K-L divergence) Let $h=h(0)$. Condition (b) in step 2 implies that
there exists a constant $c>0$ such that
$\inf_{y:|y-y_0|\le h^{1/\beta}}p_0(y)\ge
c$ for $h$ small enough. For any $i\neq j$,
\begin{align*}
&\int\log \frac{p_0(y)}{p(y|x_j)}p_0(y)dy\\
=&-\int_{y_0-h^{1/\beta}}^{y_0+h^{1/\beta}}\log \left(
1+\frac{hK((y-y_0)/h^{1/\beta})}{p_0(y)}\right)p_0(y)dy\\
 &\qquad-\int_{-y_0-h^{1/\beta}}^{-y_0+h^{1/\beta}}\log
\left(1-\frac{hK((y+y_0)/h^{1/\beta})}{p_0(y)}\right)p_0(y)dy\\
\le &
-\int_{y_0-h^{1/\beta}}^{y_0+h^{1/\beta}}
\left(\frac{hK((y-y_0)/h^{1/\beta})}{p_0(y)}-
\frac{h^2K^2((y-y_0)/h^{1/\beta})}{p_0^2(y)}\right)p(y)dy \\
&\qquad
+\int_{-y_0-h^{1/\beta}}^{-y_0+h^{1/\beta}}
\left(\frac{hK((y+y_0)/h^{1/\beta})}{p_0(y)}+
\frac{h^2K^2((y+y_0)/h^{1/\beta})}{p_0^2(y)}\right)p(y)dy\\
\le &\frac{2}{c}h^{2+\frac{1}{\beta}}\int_{-1}^1 K^2(u)du
=C h^{2+\frac{1}{\beta}}.
\end{align*}
As a result
$$KL(P_i||P_j)\le C h^{2+\frac{1}{\beta}} w^d.$$
\item Using the generalized Fano's lemma
(see also Tsybakov (2009, Chapter 2)):
\begin{align}
  \inf_{\hat C}\sup_{P}\mathbb E_P \sup_x \mu(\hat C(x)\triangle C(x))
  \ge  \frac{h}{2}\left(1-\frac{C n h^{2+\frac{1}{\beta}} w^d
  +\log 2}{-d\log w}\right),
\end{align}
where the supremum is over all $P$ such that $p(y|x)$ is Lipschtiz in $x$
in sup-norm sense, and $p(y|x)$ is H\"{o}lder smooth of order $\beta$.

Choosing $h=w=c (\log n / n )^{1/(d+2+1/\beta)}$ with constant $c$
small enough, we have
$$\inf_{\hat C}\sup_{P}\mathbb E_P \sup_x \mu(\hat C(x)\triangle C(x))
\ge c' \left(\frac{\log n }{n}\right)^{\frac{1}{d+2+\frac{1}{\beta}}}.$$

Note that the choice $h\asymp w$ is required by the condition
$$h K(0)=p(y_0|x_j)-p(y_0|x_j+w/2)\le L w.$$
\end{enumerate}

%

\bibliographystyle{apa-good}

\end{document}